\def\UseBibLatex{1}%

\documentclass[12pt]{article}%

\usepackage[cm]{fullpage}

\usepackage[usenames,dvipsnames,svgnames,table]{xcolor}
\usepackage{bm}%
\usepackage{paralist}%
\usepackage{mleftright}%
\usepackage{tabto}%
\usepackage{xifthen}%
\usepackage{stmaryrd}%

\usepackage{xspace}%
\usepackage{ifpdf}%
\usepackage{graphicx}%
\usepackage[normalem]{ulem} %
\usepackage{amsmath}%
\usepackage{amssymb}%
\usepackage[amsmath,thmmarks]{ntheorem}%
\theoremseparator{.}%

\usepackage{euscript}
\usepackage{mathrsfs}

\usepackage{verbatim}%

\usepackage{titlesec}%
\titlelabel{\thetitle. }%

\usepackage{bbm}%

   \usepackage{ifluatex}
   \usepackage{ifxetex}
   
   \ifluatex 
      \usepackage{fontspec}
      \usepackage[utf8]{luainputenc}
   \else
       \ifxetex 
          \usepackage{fontspec}
       \else
          \usepackage[T1]{fontenc}
          \usepackage[utf8]{inputenc}
       \fi
   \fi

\providecommand{\BibLatexMode}[1]{}
\providecommand{\BibTexMode}[1]{#1}

\ifx\UseBibLatex\undefined%
  \renewcommand{\BibLatexMode}[1]{}
  \renewcommand{\BibTexMode}[1]{#1}
\else
  \renewcommand{\BibLatexMode}[1]{#1}
  \renewcommand{\BibTexMode}[1]{}
\fi

\BibLatexMode{%
   \usepackage[bibencoding=ascii,style=alphabetic,backend=biber]{biblatex}%
   \usepackage{sariel_biblatex}%
}

\newcommand{\kentdelete}[1]{}

\newcommand{\IfPrinterVer}[2]{#2}%
\providecommand{\Mh}[1]{{{#1}}}%

\IfFileExists{.latex_printer_friendly}{\def\GenPrinterVer{1}}{}%

\ifx\GenPrinterVer\undefined
   \IfFileExists{.latex_color}{\def\GenColorMath{1}}{} 
\else
   \renewcommand{\IfPrinterVer}[2]{#1}%
\fi

\ifx\GenColorMath\undefined
\else
   \renewcommand{\Mh}[1]{{\textcolor{red}{#1}}}%
\fi

\IfPrinterVer{%
   \usepackage{hyperref}%
}{%
   \usepackage{hyperref}%
   \hypersetup{%
      breaklinks,%
      ocgcolorlinks, colorlinks=true,%
      urlcolor=[rgb]{0.25,0.0,0.0},%
      linkcolor=[rgb]{0.5,0.0,0.0},%
      citecolor=[rgb]{0,0.2,0.445},%
      filecolor=[rgb]{0,0,0.4},
      anchorcolor=[rgb]={0.0,0.1,0.2}%
   }
}

\newcommand{\etal}{\textit{et~al.}\xspace}

\newlength{\savedparindent}
\newcommand{\SaveIndent}{\setlength{\savedparindent}{\parindent}}
\newcommand{\RestoreIndent}{\setlength{\parindent}{\savedparindent}}

\definecolor{blue25}{rgb}{0, 0, 11}
\newcommand{\emphic}[2]{%
  \textcolor{blue25}{%
    \textbf{\emph{#1}}}%
  \index{#2}}

\ifx\PDFBlackWhite\undefined
\else
\renewcommand{\emphic}[2]{\textbf{\emph{#1}}}
\fi

\newcommand{\emphi}[1]{\emphic{#1}{#1}}

\newcommand{\cardin}[1]{\left| {#1} \right|}%
\newcommand{\ceil}[1]{\left\lceil {#1} \right\rceil}

\newcommand{\pth}[1]{\mleft({#1}\mright)}

\newcommand{\Set}[2]{\left\{ #1 \;\middle\vert\; #2 \right\}}
\newcommand{\pbrc}[1]{\mleft[ {#1} \mright]}

\newcommand{\remove}[1]{}

\newtheorem{theorem}{Theorem}[section]%
\newtheorem{lemma}[theorem]{Lemma}%
\newtheorem*{restate*}[theorem]{Restatement of }%
\newtheorem{corollary}[theorem]{Corollary}%

\newtheorem{observation}[theorem]{Observation}%

\newcommand{\myqedsymbol}{\rule{2mm}{2mm}}

\theoremstyle{remark}%
\theoremheaderfont{\sf}%
\theorembodyfont{\upshape}%
\newtheorem{defn}[theorem]{Definition}

\theoremheaderfont{\em}%
\theorembodyfont{\upshape}%
\theoremstyle{nonumberplain}%
\theoremseparator{}%
\theoremsymbol{\myqedsymbol}%
\newtheorem{proof}{Proof:}%

\numberwithin{figure}{section}%
\numberwithin{table}{section}%
\numberwithin{equation}{section}%

\newcommand{\HLinkSuffix}[3]{\hyperref[#2]{#1\ref*{#2}{#3}}}
\newcommand{\HLinkShort}[2]{\hyperref[#2]{#1\ref*{#2}}}
\newcommand{\HLink}[2]{\hyperref[#2]{#1~\ref*{#2}}}
\newcommand{\HLinkPage}[2]{\hyperref[#2]{#1~\ref*{#2}%
    $_\text{p\pageref{#2}}$}}

\newcommand{\eqrefpar}[1]{\hyperref[equation:#1]{(\ref*{equation:#1})}} %

\newcommand{\seclab}[1]{\label{sec:#1}} %
\newcommand{\secref}[1]{\HLink{Section}{sec:#1}} %

\providecommand{\deflab}[1]{\label{def:#1}}
\newcommand{\defref}[1]{\HLink{Definition}{def:#1}}

\newcommand{\lemlab}[1]{\label{lemma:#1}}
\newcommand{\lemref}[1]{\HLink{Lemma}{lemma:#1}}

\newcommand{\obslab}[1]{\label{observation:#1}}
\newcommand{\obsref}[1]{\HLink{Observation}{observation:#1}}

\newcommand{\thmlab}[1]{{\label{theo:#1}}}
\newcommand{\thmref}[1]{\HLink{Theorem}{theo:#1}}

\providecommand{\Mh}[1]{{#1}}

\newcommand{\Left}{\ell}%
\newcommand{\Right}{r}
\newcommand{\Int}{I}

\newcommand{\minksum}{\oplus}

\renewcommand{\th}{th\xspace}

\newcommand{\obj}{\Mh{f}}%
\newcommand{\objA}{\Mh{g}}%
\newcommand{\objB}{\Mh{h}}%
\newcommand{\ObjSet}{{\Mh{\mathcal{U}}}}%
\newcommand{\ObjSetA}{\Mh{\mathcal{V}}}%
\newcommand{\ObjSetB}{\Mh{\mathcal{H}}}%

\newcommand{\Cover}{\Mh{\mathcal{C}}}%

\newcommand{\bNotation}[1]{\Mh{#1}}%
\newcommand{\ball}{\bNotation{b}}%
\newcommand{\ballY}[2]{\Mh{\mathbbm{b}}\pth{#1, #2}}%
\newcommand{\BallSet}{\Mh{\mathcal{B}}}%

\newcommand{\SetA}{\Mh{X}}%
\newcommand{\SetB}{\Mh{Y}}%
\newcommand{\SetC}{\Mh{U}}

\newcommand{\cen}{\Mh{c}}

\newcommand{\diamX}[1]{\operatorname{\Mh{diam}}\pth{#1}}%

\newcommand{\Lines}{\mathcal{L}} %
\newcommand{\lineA}{\ell_1}      %
\newcommand{\lineB}{\ell_2}      %

\newcommand{\seg}{s}%
\renewcommand{\Re}{{\mathbb{R}}}

\newcommand{\naturalnumbers}{\mathbb{N}} %

\newcommand{\gradC}{\mathbbm{d}} %

\newcommand{\gradY}[2]{\gradC_{#1}\pth{#2}}

\newcommand{\clusters}{\Mh{\mathcal{W}}} %
\newcommand{\distY}[2]{\mleft\| #1 - #2 \mright\|}
\newcommand{\normX}[1]{\left\| #1 \right\|}

\newcommand{\cDensity}{\Mh{\rho}} %
\newcommand{\densityOp}{\Mh{\mathop{\mathrm{density}}}}%
\newcommand{\PointDec}[1]{\Mh{#1}}

\newcommand{\pnt}{\PointDec{p}}%
\newcommand{\pntA}{\PointDec{q}}%
\newcommand{\pntB}{\PointDec{u}} %
\newcommand{\pntC}{\PointDec{v}}

\newcommand{\IntSet}{\mathcal{I}} %

\newcommand{\eps}{\Mh{\varepsilon}}%

\newcommand{\SepSet}{\Mh{Z}}%

\newcommand{\Vertices}{\Mh{V}}%
\newcommand{\VerticesA}{\Mh{U}}

\newcommand{\SetL}{\Mh{L}}%
\newcommand{\SetR}{\Mh{R}}

\newcommand{\VerticesX}[1]{\Mh{V}\pth{#1}}%

\newcommand{\Edges}{\Mh{E}}
\newcommand{\EdgesX}[1]{\Edges\pth{#1}}

\DefineNamedColor{named}{ColorComplexityClass}{cmyk}{0.64,0.0,0.95,0.80}

\newcommand{\poly}{\operatorname{poly}}%

\newcommand{\cSD}{\sigma}

\newcommand{\si}[1]{#1}%

\newcommand{\Graph}{\graph}%

\newcommand{\GInduced}[1]{\graph_{|{#1}}}

\newcommand{\atgen}{\symbol{'100}}%

\providecommand{\tildegen}{{\protect\raisebox{-0.1cm}
    {\symbol{'176}\hspace{-0.01cm}}}}

\newcommand{\KentThanks}%
{%
  \thanks{%
    Department of Computer Science; %
    University of Illinois; %
    201 N. Goodwin Avenue; %
    Urbana, IL, 61801, USA; %
    {\tt quanrud2\atgen{}illinois.edu}; %
    {\tt\href%
      {http://illinois.edu/\string~quanrud2/}%
      {http://illinois.edu/\tildegen{}quanrud2/}%
    }%
    . %
  }%
}%
\newcommand{\SarielThanks}[1][]{%
  \thanks{Department of Computer Science; %
    University of Illinois; %
    201 N. Goodwin Avenue; %
    Urbana, IL, 61801, USA; %
    {\tt sariel\atgen{}illinois.edu}; %
    {\tt \url{http://sarielhp.org/}}. #1}} %

\newcommand{\distSet}[2]{d\pth{#1,#2}}

\DeclareMathAlphabet{\mathantt}{OT1}{antt}{li}{it}
\DeclareMathAlphabet{\mathpzc}{OT1}{pzc}{m}{it}

\DeclareMathAlphabet{\mathcalligra}{T1}{calligra}{m}{n}

\newcommand{\lenX}[1]{\normX{#1}}

\newcommand{\pleftX}[1]{\Left\pth{#1}}
\newcommand{\prightX}[1]{\Right\pth{#1}}

\newcommand{\xLof}[1]{x_\Left\pth{#1}}%
\newcommand{\hL}[1]{y_\Left\pth{#1}}%

\newcommand{\xR}[1]{x_\Right\pth{#1}}%
\newcommand{\hR}[1]{y_\Right\pth{#1}}%

\newcommand{\kSD}{k}

\newcommand{\exSize}{\Mh{\lambda}}%

\newcommand{\IncludeGraphics}[2][]{%
  \typeout{}%
  \typeout{Graphics: #2}%
  \typeout{\ includegraphics[#1]{#2}}%
  \includegraphics[#1]{#2}
  \typeout{}%
}

\newcommand{\defGraph}{\graph = (\Vertices,\Edges)}

\newcommand{\class}{\Mh{\mathcal{C}}}
\newcommand{\GraphNotation}[1]{\Mh{#1}}

\newcommand{\graph}{\GraphNotation{G}}%
\newcommand{\graphA}{\GraphNotation{H}}%
\newcommand{\TheoremDefExt}[3]{%
  \expandafter\newcommand\csname bodyThm#1\endcsname{#2}
  \ifthenelse{\isempty{#3}}{%
    \begin{theorem}%
      \thmlab{#1}%
      #2
    \end{theorem}%
  }{%
    \begin{theorem}[#3]%
      \thmlab{#1}%
      #2
    \end{theorem}%
  }%
}

\newcommand{\TheoremBody}[1]{%
  \csname bodyThm#1\endcsname%
}

\newcommand{\LemmaDefExt}[3]{%
  \expandafter\newcommand\csname bodyLemBody#1\endcsname{#2}
  \ifthenelse{\isempty{#3}}{%
    \begin{lemma}%
      \lemlab{#1}%
      #2
    \end{lemma}%
  }{%
    \begin{lemma}[#3]%
      \lemlab{#1}%
      #2
    \end{lemma}%
  }%
}

\newcommand{\LemmaBody}[1]{%
  \csname bodyLemBody#1\endcsname%
}

\definecolor{sarielChangeColor}{rgb}{0.3,0.451,0.055}%
\IfFileExists{sariel_computer.sty}{\def\sarielComp{1}}{}
\ifx\sarielComp\undefined%
\newcommand{\SarielComp}[1]{}
\newcommand{\NotSarielComp}[1]{#1}%
\else
\newcommand{\SarielComp}[1]{#1}%
\newcommand{\NotSarielComp}[1]{}%
\fi

\SarielComp{%
  \ifx\colorMath\undefined%
  \else
  \DefineNamedColor{named}{ColorMath}{rgb}{0.5,0.2,0}
  \renewcommand{\Mh}[1]{{\textcolor{ColorMath}{#1}}}
  \fi
}

\IfPrinterVer{%
  \definecolor{blue25}{rgb}{0,0,0}
  \DefineNamedColor{named}{RedViolet} {rgb}{0,0,0}%
}{%
  \definecolor{blue25}{rgb}{0,0,0.7}
  \DefineNamedColor{named}{RedViolet} {cmyk}{0.07,0.90,0,0.34}
}%

\BibLatexMode{%
}

\begin{document}

\title{Notes on Approximation Algorithms for Polynomial-Expansion and
   Low-Density Graphs%
   \thanks{Work on this paper was partially supported by a NSF AF
      awards CCF-1421231, and 
      CCF-1217462. %
   }%
}%

\author{%
   Sariel Har-Peled%
   \SarielThanks{}%
   \and%
   Kent Quanrud%
   \KentThanks{}%
}%

\maketitle

\begin{abstract}
    This write-up contains some minor results and notes related to our
    work \cite{hq-aaldg-15-arxiv}. In particular, it shows the
    following:
    \begin{compactenum}[(A)]
        \item In \secref{proof-p:e:separator} we show that a graph
        with polynomial expansion have sublinear separators.

        \item In \secref{proof:divisions:small:excess} we show that
        hereditary sublinear separators imply that a graph have small
        divisions.

        \item In \secref{exposed:segs}, we show a natural condition on
        a set of segments, such that they have low density. This might
        be of independent interest in trying to define a realistic
        input model for a set of segments. Unlike the previous two
        results, this is new.
    \end{compactenum}

    For context and more details, see the main paper \cite{hq-aaldg-15-arxiv}.
\end{abstract}

\section{Polynomial expansion implies sublinear %
   separators}
\seclab{proof-p:e:separator}

\begin{defn}%
    \deflab{separator}%
    Let $\graph = (\Vertices,\Edges)$ be an undirected graph.  Two
    sets $\SetA, \SetB \subseteq \Vertices$ are \emphi{separate} in
    $\graph$ if
    \begin{inparaenum}[(i)]
        \item $\SetA$ and $\SetB$ are disjoint, and
        \item there is no edge between the vertices of $\SetA$ and
        $\SetB$ in $\graph$.
    \end{inparaenum}
    A set $\SepSet \subseteq \Vertices$ is a \emphi{separator} for a
    set $\SetC \subseteq \Vertices$, if
    $\cardin{\SepSet} = o\pth{\cardin{\SetC}}$, and
    $\SetC \setminus \SepSet$ can be partitioned into two
    \emph{separate} sets $\SetA$ and $\SetB$, with
    $\cardin{\SetA} \leq (2/3)\cardin{\Vertices}$ and
    $\cardin{\SetB} \leq (2/3)\cardin{\Vertices}$\footnote{Here, the
       choice of $2/3$ is arbitrary, and any constant smaller than $1$
       is sufficient.}.
\end{defn}


\begin{theorem}[{{\cite[Theorem 2.3]{prs-semig-94}}}]
    \thmlab{excluded-shallow-minor-separator}%
    Let $\graph$ be a graph with $m$ edges and $n$ vertices, and let
    $\ell, h \in \naturalnumbers$ be two integer parameters. There is
    an $O(m n/\ell)$ time algorithm that either produces
    \begin{compactenum}[\qquad(a)]
        \item the clique $K_h$ as a $\ell \log n$-shallow minor of
        $\graph$, or
        \item a separator of size at most
        $O(n/\ell + 4 \ell h^2 \log n)$.
    \end{compactenum}
\end{theorem}

\TheoremDefExt{p:e:separator}%
{%
   Let $\class$ be a class of graphs with polynomial expansion of
   order $k$. For any graph $\graph \in \class$ with $n$ vertices and
   $m$ edges, one can compute, in
   $O\pth{m n^{1-\alpha} \log^{1-\alpha} n}$ time, a separator of size
   \begin{math}
       O \bigl( n^{1 - \alpha} \log^{1 - \alpha} n \bigr),
   \end{math}
   where $\alpha = 1/\pth{2k+2}$.%
}%
{{{\cite[Theorem 8.3]{no-gcbe2-08}}}}%

\begin{proof}
    Let $z$ be a parameter to be fixed shortly, and let
    $\ell = z / \log n$ and $cz^k/4 > \gradY{z}{\graph}$, where $c$ is
    a sufficiently large constant. Consider a $z$-shallow minor
    $\graphA$ of $\graph$ with $h = cz^k$ vertices, and observe that
    by definition, we have that
    \begin{math}
        \displaystyle%
        \cardin{\EdgesX{\graphA}}%
        \leq %
        \gradY{z}{\graph} \cardin{\VerticesX{\graphA}}%
        < %
        \frac{cz^k}{4} cz^k%
        <%
        \binom{h}{2}.
    \end{math}
    That is, the graph $H$ can not be the clique $K_h$.

    Now, by \thmref{excluded-shallow-minor-separator}, $\graph$ has a
    separator of size
    \begin{align*}
      O\pth{\Bigl.n / \ell + \ell h^2 \log n} %
      &= %
        O\pth{ %
        \frac{n \log n}{z} + \frac{z}{\log n} \cdot z^{2k} \cdot \log n
        } %
              = %
              O\pth{ %
              \frac{n \log n}{z} + z^{2k+1}} %
              = %
              O\pth{ n^{\frac{2k+1}{2k+2}} \log^{\frac{2k+1}{2k+2}} n}
    \end{align*}
    for $z = n^{1/(2k+2)} \log^{1/(2k+2)} n$. The algorithm provided
    by \thmref{excluded-shallow-minor-separator} runs in time
    \begin{math}
        O\pth{ \frac{m n }{\ell}}%
        =%
        O( \frac{m n \log n}{z} )%
        \linebreak[2]
        =%
        O(m n^{\frac{2k+1}{2k+2}} \log^{\frac{2k+1}{2k+2}} n).
    \end{math}
\end{proof}

\section{Hereditary separators imply small divisions}
\seclab{proof:divisions:small:excess}

Consider a set $\Vertices$. A \emphi{cover} of $\Vertices$ is a set
$\clusters = \Set{C_i \subseteq \Vertices}{i=1,\ldots, k \bigr.}$ such
that $\Vertices = \bigcup_{i=1}^k C_i$. A set $C_i \in \clusters$ is a
\emphi{cluster}. A cover of a graph $\defGraph$ is a cover of its
vertices. Given a cover $\clusters$, the \emph{excess} of a vertex
$v \in \Vertices$ that appears in $j$ clusters is $j-1$. The
\emphi{total excess} of the cover $\clusters$ is the sum of excesses
over all vertices in $\Vertices$.

\begin{defn}
    A cover $\Cover$ of $\graph$ is a \emphi{$\exSize$-division} if
    \begin{inparaenum}[(i)]
        \item for any two clusters $C, C' \in \Cover$, the sets
        $C \setminus C'$ and $C' \setminus C$ are separated in
        $\graph$ (i.e., there is no edge between these sets of
        vertices in $\graph$), and
        \item for all clusters $C \in \Cover$, we have
        $\cardin{C} \leq \exSize$.
    \end{inparaenum}

    A vertex $v \in \Vertices$ is an \emphi{interior vertex} of a
    cover $\clusters$ if it appears in exactly one cluster of
    $\clusters$ (and its excess is zero), and a \emphi{boundary
       vertex} otherwise. By property (i), the entire neighborhood of
    an interior vertex of a division lies in the same cluster.
\end{defn}

The property of having $\exSize$-divisions is slightly stronger than
being weakly hyperfinite.  Specifically, a graph is \emph{weakly
   hyperfinite} if there is a small subset of vertices whose removal
leaves small connected components \cite[Section
16.2]{no-s-12}. Clearly, $\exSize$-divisions also provide such a set
(i.e., the boundary vertices). The connected components induced by
removing the boundary vertices are not only small, but the
neighborhoods of these components are small as well.

As noted by Henzinger \etal \cite{hkrs-fspapg-97}, strongly sublinear
separators obtain $\exSize$-divisions with total excess $\eps n$ for
$\exSize = \poly(1/\eps)$.  Such divisions were first used by
Frederickson in planar graphs \cite{f-faspp-87}.

\begin{lemma}[\cite{hkrs-fspapg-97}]
    Let $\graph$ be a graph with $n$ vertices, such that any induced
    subgraph with $m$ vertices has a separator with
    $O(m^{\alpha} \log^\beta m)$ vertices, for some $\alpha < 1$ and
    $\beta \geq 0$.  Then, for $\eps > 0$, the graph $\graph$ has
    $\exSize$-divisions with total excess $\eps n$, where
   \begin{math}
       \exSize = O \pth{ \pth{ \eps^{-1} \log^\beta \eps^{-1}
          }^{1/(1-\alpha)}}.
   \end{math}%
\end{lemma}

\begin{proof}
    Our strategy is to break $\graph$ into smaller
    pieces. Specifically, at every step the algorithm takes the
    largest remaining piece $\GInduced{\VerticesA}$, compute a
    balanced separator $\SepSet \subseteq \VerticesA$ for it, with
    $\SetL, \SetR \subseteq \VerticesA$ being the two separated
    pieces. Specifically, we have %
    \smallskip%
    \begin{compactenum}[\qquad(i)]
        \item $\SepSet = \SetL \cap \SetR$,
        \item $\SetL \cup \SetR = \VerticesA$,
        \item $\cardin{\SetL} \leq (2/3) \cardin{\VerticesA}$ and
        $\cardin{\SetR} \leq (2/3) \cardin{\VerticesA}$ (see
        \defref{separator}),
        \item $\SetL \setminus \SepSet$ is separated from
        $\SetR \setminus \SepSet$ in $\GInduced{\VerticesA}$, and
        \item
        $\cardin{\SepSet} \leq f\pth{\cardin{\VerticesA} \bigr.}$,
        where $f(m ) \leq c m^{\alpha} \log^\beta m$, where $c$ is a
        sufficiently large constant.
    \end{compactenum}
    \smallskip%
    Now, the algorithm replaces $\GInduced{\VerticesA}$ by the two
    ``broken'' pieces $\GInduced{L}$ and $\GInduced{R}$.  The
    algorithm continues in this process until all pieces are of size
    smaller than $b$ (and by construction, of size at least, say,
    $b/4$), where $b$ is some parameter to be specified shortly.  This
    generates a natural binary separator tree, where the final pieces
    of the division are the leafs.

    Let $N_i = (3/4)^i n$, for $i=0, \ldots, h=\ceil{ \log_{4/3} n}$.
    A piece $\graph_\VerticesA$ is at \emph{level $i$} if
    $N_{i+1} < \cardin{\VerticesA} \leq N_i$.
    Consider such a subproblem at node $y$, which is at level $i$ with
    $\nu$ vertices. The total size of the subproblems of its two
    children is $\leq \nu + 2f\pth{\nu}$ (here, somewhat confusingly,
    we count the separator vertices as new, in both subproblems --
    this makes the following argument somewhat easier).  Importantly,
    each of the subproblems is of size
    $\leq (2/3)\nu + f\pth{\nu} \leq (3/4) \nu$, implying that both
    subproblems are in strictly lower level.  As such, the fraction of
    the new vertices created as subproblems move from the $i$\th level
    to the next is bounded by
    \begin{align*}
      \nu + 2f\pth{\nu}%
      \leq%
      \nu + 2c \nu^\alpha \log^\beta \nu%
      =%
      \pth{1 + \frac{2c \log^\beta \nu}{\nu^{1- \alpha } }} \nu%
      \leq%
      \gamma_i \nu,
    \end{align*}
    for
    \begin{math}
        \displaystyle%
        \gamma_i %
        = %
        {1 + 2\frac{c \log^\beta N_{i+1}}{\pth{N_{i+1}}^{1 - \alpha }
           }}.
    \end{math}
    In particular, the total number of vertices in the $k$\th level is
    at most $\Delta_k n$, where
    \begin{align*}
      \Delta_k%
      &=%
        \prod_{j=0}^{k-1} \gamma_j%
        \leq%
        \prod_{j=0}^{k-1} \exp \pth{ 2 \frac{c \log^\beta
        N_{j+1}}{\pth{N_{j+1}}^{1 - \alpha} }} %
      =%
      \exp \pth{ \sum_{j=0}^{k-1} \frac{2 c \log^\beta
      N_{j+1}}{\pth{N_{j+1}}^{1 - \alpha} }}%
      \leq%
      \exp \pth{ \frac{ c' \log^\beta N_{k}}{\pth{N_{k}}^{1 - \alpha}
      }}%
      \\%
      &\leq%
        1 + \frac{ 2 c' \log^\beta N_{k}}{\pth{N_{k}}^{1 - \alpha} },
    \end{align*}
    since the summation behaves like an increasing geometric series,
    and $c'$ is a constant that depends on $c$. The last step follows
    as $e^x \leq 1+2x$, for $0 \leq x \leq 1/2$. In particular,
    because of the double counting of the separator vertices, the
    total number of marked vertices in the first $k$ levels is bounded
    by $n\pth{\Delta_k - 1}$. As such, we need that
    $\Delta_k - 1\leq\eps$. This is equivalent to
    \begin{align*}
      \frac{ 2 c' \log^\beta N_{k}}{\pth{N_{k}}^{1 - \alpha} } \leq \eps
      \iff%
      \frac{ 2 c' }{\eps } \leq \frac{\pth{N_{k}}^{1 -
      \alpha}}{\log^\beta N_k},
    \end{align*}
    which holds if
    $N_k \geq \pth{ c'' \eps^{-1} \log^\beta \eps^{-1}
    }^{1/(1-\alpha)}$, where $c''$ is a sufficiently large
    constant. In particular, setting $b$ to (say) twice this threshold
    implies the claim.
\end{proof}


\section{On exposed sets of segments and their %
  density}
\seclab{exposed:segs}

Let $\cSD > 0$ be a fixed parameter. We say that an object $\obj$
\emphi{$\cSD$-shadows} (or simply \emphi{shadows}) another object
$\objA$ if
\begin{align*}
    \max_{\pntA \in \objA}\distSet{\pntA}{\obj} &\leq \cSD \cdot
    \diamX{\objA},
\end{align*}
where
\begin{math}
    \distSet{\pntA}{\obj} = \min_{\pntB \in \obj} \distY
    {\pntA}{\pntB}.
\end{math}
Equivalently, $\obj$ $\cSD$-shadows $\objA$ $\iff$
\begin{math}
    \objA%
    \subseteq %
    \obj \minksum \ballY{0}{\cSD \cdot \diamX{\objA}}.%
\end{math}
Here,
\begin{math}
    \SetA \minksum \SetB = \Set{\pntA + \pntB}{\pntA \in \SetA, \pntB
       \in \SetB}
\end{math}
denotes the \emphi{Minkowski sum} of $\SetA$ and $\SetB$. A set of
objects $\ObjSet$ is \emphi{$\cSD$-exposed} if no object in $\ObjSet$
$\cSD$-shadows another object in $\ObjSet$.

\begin{observation}%
    \obslab{containing-sets-shadow}%
    Let $\obj$ and $\objA$ be two objects and $\cSD \geq 0$. If
    $\obj \subseteq \objA$, then $\objA$ $\sigma$-shadows $\obj$.
\end{observation}

\begin{figure}
    \centerline{%
       \begin{tabular}{c%
         cc}
         \IncludeGraphics[page=1]{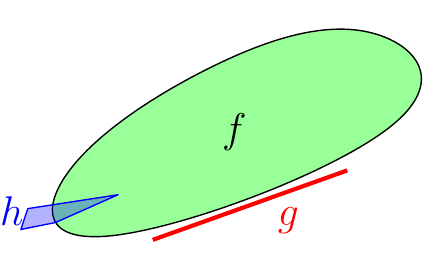}
         &%
           \IncludeGraphics[page=2]{figs/shadow}
         &%
           \IncludeGraphics[page=3]{figs/shadow}
         \\[-.4cm]
         (A) & (B) & (C)
       \end{tabular}%
    }
    \caption{(A) Objects $\obj, \objA, \objB$. (B) $\obj$
       $1/8$-shadows $\objA$. (C) $\obj$ does not $1/8$-shadow
       $\objB$.}
\end{figure}

\subsection{On the density of exposed segments}

\subsubsection{Intervals in $\Re$}

Following the above, interval $\Int = [\Left,\Right]$
\emphi{$\sigma$-exposes} $\Int' = [\Left',\Right']$, if $\Int'$ is not
contained in the interval
$\pbrc{\Left - \sigma\lenX{J}, \Right+ \sigma \lenX{J}\bigr.}$, where
$\lenX{\Int'} = \Left'-\Right'$ denotes the length of $\Int'$.


\bigskip
\noindent
\begin{minipage}{0.7\linewidth}
    \begin{lemma}%
        \lemlab{n:s:intervals}%
        Let $I = [\Left,\Right]$ and $I' = [\Left',\Right']$ be two
        overlapping intervals on the real line. If $I$ and $I'$
        $\cSD$-expose each other, then
        $\cardin{\Left' - \Left} \geq \cSD \lenX{I}$ and
        $\cardin{\Right' - \Right} \geq \cSD \lenX{I'}$.
    \end{lemma}
\end{minipage}\hfill
\begin{minipage}{0.25\linewidth}
    \IncludeGraphics[width=0.99\linewidth]{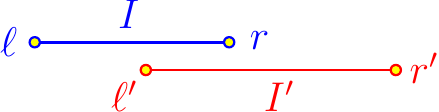}
\end{minipage}

\begin{proof}
    Without loss of generality, assume that $\Left \leq \Left'$. Since
    $I$ and $I'$ are overlapping, we have $\Left' \leq \Right$.
    Furthermore, if $\Right' \leq \Right$, then $I \subseteq I'$ and
    by \obsref{containing-sets-shadow} the interval $I$ $\cSD$-shadows
    $I'$.  So it must be that $\Right' > \Right$.  Since the left
    endpoint $\Left'$ of $I'$ is contained in $\Int$, $\Int$ does not
    $\cSD$-shadow $\Int'$ only if $\Right'$ extends at least
    $\cSD \lenX{\Int'}$ past $\Right$. Similarly, if $\Int'$ does not
    $\cSD$-shadow $\Int$, then $\Left' - \Left \geq \cSD \lenX{\Int}$.
\end{proof}

\begin{lemma}
    \lemlab{int:low:density}%
    Let $\IntSet$ be a set of intervals all covering a common point
    $\pnt$. If $\IntSet$ is a $\cSD$-exposed set of intervals, then
    $\cardin{\IntSet} = O\pth{1/\cSD^{2}}$.
\end{lemma}

\begin{proof}
    Let the $i$\th interval of $\IntSet$ be
    $\Int_i = [\Left_i, \Right_i]$, for $i=1,\ldots, n$.  Furthermore,
    assume that $\Left_1 \leq \Left_2 \leq \cdots \Left_n$.  By
    Dilworth's theorem, there exists a subsequence
    $i_1 < i_2 < \cdots < i_k$ with $k \geq \sqrt{n}$, such that
    either
    $\Right_{i_1} \leq \Right_{i_2} \leq \cdots \leq \Right_{i_k}$ or
    $\Right_{i_1} \geq \Right_{i_2} \geq \cdots \geq \Right_{i_k}$.
    The later possibility implies that $I_{i_1} \subseteq I_{i_2}$,
    which contradicts the assumption that $\IntSet$ is $\cSD$-exposed.
    Assume, without loss of generality, that for at least half the
    intervals in this sequence, we have
    $\cardin{\Left(I)} \geq \Right(I)$, and let
    $\Int_1', \ldots, \Int_{k/2}'$ be the resulting subsequence
    restricted to these intervals, where
    $\Int_1' = \pbrc{\Left_i', \Right_i'}$ for all $i$. (The other
    case is handled by symmetric argument.)

    We have
    $\Left_1' \leq \ldots \Left_{k/2}' \leq 0 \leq \Right_1' \leq
    \cdots \leq \Right_{k/2}'$.
    By \lemref{n:s:intervals}, we have that for any $i$, we have
    $\Right_i' - \Right_{i-1}' \geq \cSD \lenX{ \Int_{i-1}'} \geq \cSD
    \cardin{\Left_{i-1}'}$.
    Summing this inequality for $i=2,\ldots, t$, we have
    \begin{align*}
        \Right_t'%
        >%
        \Right_t' - \Right_{1}'%
        \geq%
        \cSD \sum_{i=1}^{t-1}\cardin{\Left_{i}'}%
        \geq%
        \cSD \cdot (t-1) \cardin{\Left_{t}}%
        > %
        \cardin{\Left_{t}},
    \end{align*}
    for $t= \ceil{1/\cSD}+2$, which is a contradiction. We conclude
    that $\sqrt{n}/2 \leq k/2 \leq \ceil{1/\cSD}+2$, which readily
    implies the claim.
\end{proof}

\subsubsection{Line segments through a point}

\begin{lemma}%
    \lemlab{small:angles}%
    Let $\Lines$ be a set of segments in $\Re^d$, and $\cSD > 0$,
    $\theta \in (0,\pi/2)$ be parameters. Furthermore, assume that
    \begin{inparaenum}[(i)]
        \item $\Lines$ is $\cSD$-exposed,
        \item $\bigcap_{\seg \in \Lines} \seg \neq \emptyset$,
        \item for all pairs $\lineA, \lineB \in \Lines$, the angle
        between $\lineA$ and $\lineB$ is at most $\theta$, and
        \item $\sin \theta \leq \frac{\cSD}{4}$.
    \end{inparaenum}
    Then $\cardin{\Lines} =O(1/ \cSD^{2})$.
\end{lemma}
\begin{proof}
    Without loss of generality, we assume that the lines intersect at
    the origin and the angle between any line and the $x$-axis is at
    most $\theta$. For each $\seg \in \Lines$, let $\pleftX{\seg}$ be
    the left endpoint of $\seg$, let $\xLof{\seg}$ be the
    $x$-coordinate of $\pleftX{\seg}$, and let $\hL{\seg}$ the
    distance from $\pleftX{\seg}$ to the $x$-axis.  Similarly we
    define $\prightX{\seg}$, $\xR{\seg}$, and $\hR{\seg}$ with respect
    to the right endpoint. For $\seg \in \Lines$, let
    $\Int_\seg = \pbrc{\xLof{\seg},\xR{\seg}}$ be the projection of
    $\seg$ onto the $x$-axis, and let
    $\IntSet_\Lines = \Set{\Int_\seg}{ \seg \in \Lines}$.

    We claim that $\IntSet$ is $(\cSD/4)$-exposed. Indeed, suppose
    that there are two segments $\seg, \seg'$, such that
    $\Int_\seg = \pbrc{\Left, \Right}$ is $\cSD/4$-shadowing
    $\Int_{\seg'} = \pbrc{\Left', \Right'}$.  We define the following
    sequence of points: %
    \SaveIndent%

    \smallskip
    \noindent%
    \begin{minipage}{0.6\linewidth}
        \RestoreIndent%
        {\medskip}%
        \begin{compactenum}[\quad(i)]
            \item $\pnt$ is any point on $\seg'$,

            \item $\pntA$ is the projection of $\pnt$ into
            $\Int_{\seg'}$,

            \item $\pntB$ is any point in $\Int_\seg$ that is in
            distance at most $(\cSD/4) \lenX{\Int_{\seg'}}$ from
            $\pntA$, and it is closer to the origin than $\pntB$, and

            \item $\pntC$ is the point on $\seg$ whose projection on
            $\Int_\seg$ is $\pntB$.
        \end{compactenum}
        {\smallskip}%
        See the figure on the right.
    \end{minipage}
    \begin{minipage}{0.39\linewidth}
        \hfill
        \IncludeGraphics[width=0.95\linewidth]%
        {\si{figs/seg_shadow}}%
    \end{minipage}%

    \smallskip%
    \noindent%

    We now have that
    $\distY{\pnt}{\pntA} \leq \lenX{\seg'} \sin \theta \leq (\cSD/4)
    \lenX{\seg'}$.
    Similarly, as $\pntB$ is closer to the origin than $\pntA$, we
    have that $\distY{\pntB}{\pntC} \leq (\cSD/4) \lenX{\seg'}$.
    Also, since
    \begin{math}
        \distY{\pntA}{\pntB} \leq (\cSD/4) \lenX{\Int_{\seg'}} \leq
        (\cSD/4) \lenX{\seg'},
    \end{math}
    we have by the triangle inequality that
    \begin{math}
        \distY{\pnt}{\pntC}%
        \leq %
        \distY{\pnt}{\pntA} + \distY{\pntA}{\pntB} +
        \distY{\pntB}{\pntC}%
        \leq%
        \cSD \lenX{\seg'},
    \end{math}
    which implies that $\seg'$ is $\cSD$-shadowed by $\seg$, a
    contradiction.

    Now, \lemref{int:low:density} implies that
    \begin{math}
        \cardin{\Lines} = \cardin{\IntSet_\Lines} = O\pth{1/\cSD^2}.
    \end{math}
\end{proof}

\begin{lemma}%
    \lemlab{segs:through:p}%
    Let $\Lines$ be a set of segments in $\Re^d$ and $\cSD \in (0,1)$
    a fixed parameter, such that
    \begin{inparaenum}[(i)]
        \item $\Lines$ is $\sigma$-exposed, and
        \item $\bigcap_{\seg \in \Lines} \seg \neq \emptyset$.
    \end{inparaenum}
    Then $\cardin{\Lines} = O\pth{1/\cSD^{d+2}}$.
\end{lemma}
\begin{proof}
    Partition $\Lines$ into $O(\cSD^{-d})$ clusters such that any two
    lines in the same cluster forms an angle $\leq \cSD /4$.  By
    \lemref{small:angles}, each cluster contains at most
    $O\pth{1/\cSD^{2}}$ segments, and the claim follows.
\end{proof}

\subsubsection{Large segments all intersecting a common ball}

\begin{lemma}
    \lemlab{large:segments}%
    Let $\ball$ be a ball of radius $r$, and let $\Lines$ be a set of
    segments both in $\Re^d$.  Furthermore, assume that (i) $\Lines$
    is $\cSD$-exposed, (ii) all the segments of $\Lines$ intersect
    $\ball$, and (iii) they are all of length $\geq r$. Then, we have
    $\cardin{\Lines} = O\pth{1/\cSD^{2d+2}}$.
\end{lemma}
\begin{proof}
    Let $\BallSet$ be a set of $O\pth{\cSD^{-e}}$ balls of radius
    $\cSD r/4$, that cover $\ball$.  For each $\seg \in \Lines$, pick
    a small ball $\ball_{\seg} \in \BallSet$ intersecting $\seg$, and
    translate $\seg$ by at most $\cSD r/4$ so that it passes through
    the center of $\ball_{\seg}$. For $\seg \in \Lines$, let $\seg'$
    denote the translated segment, and let
    $\Lines' = \Set{\seg'}{\seg \in \Lines}$.

    Since $\Lines$ is $\cSD$-exposed, and the length of each segment
    of $\Lines$ is at least $r$, it follows that $\Lines'$ is
    $\cSD/2$-exposed, as can be easily verified.

    Now, for every ball $\ballY{\cen}{\cSD r/4} \in \BallSet$,
    consider the segment of segments $\Lines'(\cen)$ that passes
    through $\cen$. By \lemref{segs:through:p}, we have
    $\cardin{\Lines'(\cen)} = O\pth{1/\cSD^{d+2}}$. This implies that
    \begin{math}
        \cardin{\Lines} = \cardin{\Lines'}%
        =%
        O\pth{\cardin{\BallSet} /\cSD^{d+2} }%
        =%
        O\pth{1 /\cSD^{2d+2} }.
    \end{math}
\end{proof}

\subsubsection{Putting things together}

\begin{lemma}
    Let $\Lines$ be a set of segments in $\Re^d$ and $\cSD > 0$ a
    fixed parameter. If $\Lines$ is $\cSD$-exposed, then $\Lines$ has
    density $O(\cSD^{-2d-2})$.
\end{lemma}
\begin{proof}
    Consider any ball $\ballY{\cen}{r}$ in $\Re^d$. By
    \lemref{large:segments}, there could be at most $O(\cSD^{-2d-2})$
    segments of length $\geq 2r$ of $\Lines$ intersecting it, and the
    result follows.
\end{proof}

\subsection{On $(\cSD, k)$-shadowing}

A set of objects $\ObjSet$ in $\Re^d$ is \emphi{$(\cSD,\kSD)$-exposed}
if each object $\obj \in \ObjSet$ is $\cSD$-shadowed by at most $k$
other objects in $\ObjSet$.
\begin{lemma}
    Let $\cSD > 0$ be a fixed parameter and $\ObjSetA$ a set of
    objects, such that for any subset $\ObjSetB \subseteq \ObjSetA$
    that is $\cSD$-exposed, we have that
    \begin{math}
        \densityOp(\ObjSetB) \leq \cDensity.
    \end{math}
    If $\ObjSetA$ is $(\cSD, k)$-exposed, then
    $\densityOp(\ObjSetA) \leq (2k + 1) \cDensity$.
\end{lemma}

\begin{proof}
    We create a graph $\Graph$ over $\ObjSetA$, with an edge between
    two objects $\objA, \objB \in \ObjSetA$ if one shadows the
    other. By assumption, the average degree in $\Graph$ is bounded by
    $2k$, and in particular the graph is $2k$-degenerate and can be
    partitioned into $2k+1$ independent sets. Every independent set is
    $\cSD$-exposed, and by assumption has density $\leq \cDensity$.
    Since density is subadditive under unions, $\ObjSetA$ has density
    at most $(2k+1) \cDensity$.
\end{proof}

\begin{corollary}
    Let $\Lines$ be a set of segments in $\Re^d$ that
    $(\cSD, \kSD)$-exposed.  Then $\Lines$ has density
    $O\pth{k \cSD^{-4}}$.
\end{corollary}

\BibTexMode{%
}

\BibLatexMode{\printbibliography}

\end{document}